\documentclass[12pt]{article}
\usepackage{amsthm,amssymb,amsmath}
\usepackage{bm,bbm}
\usepackage{mathtools,etoolbox,xcolor}
\usepackage{enumitem}
\usepackage[margin=1in]{geometry}
\usepackage{authblk}
\usepackage[numbers]{natbib}
\usepackage{booktabs,siunitx}
\usepackage{caption,subcaption}
\usepackage{tikz}
\newtheorem{theorem}{Theorem}

\newtheorem{proposition}[theorem]{Proposition}

\newcommand{\E}{\mathbbm{E}}

\newcommand{\n}{n}
\newcommand{\Y}{Y}

\newcommand{\se}{\hat\sigma}

\newcommand{\w}{w}

\newcommand{\TLR}{W}

\newcommand{\munull}{\mu_{0}}

\newcommand{\mumle}{\hat\mu_{MLE}}
\newcommand{\tausqmle}{\hat\tau^2_{MLE}}
\newcommand{\tausqmlered}{\hat\tau^2_{MLE,R}}

\graphicspath{{./figs/}} 
\newtoggle{commenttoggle}
\togglefalse{commenttoggle}
\mathtoolsset{showonlyrefs,showmanualtags}

\title{Bartlett adjustment for Gaussian random effects meta-analysis }
\author[1]{Haben Michael}
\affil[1]{University of Massachusetts, Amherst, MA (hmichael@math.umass.edu)}
\date{}                   
\begin{document}
\maketitle

Abstract. 
Meta-analyses are often based on too few studies to justify the
asymptotic methods underlying the statistical procedures applied to
them. We consider higher-order asymptotics as a remedy.  We derive the
Bartlett correction for the idealized Gaussian case, correcting the
formula currently appearing in the literature.
\\

Meta-analysis is a technique for synthesizing a body of related scientific results into a summary conclusion. Suppose each of $\n$ studies, all measuring the same underlying phenomenon,  produces an estimate $\Y_i$ and standard error $\se_i$. In the likelihood ratio approach to meta-analysis \citep{hardy:96},
the data are assumed to be independent and conditionally Gaussian,
\begin{align}
  \Y_i \mid \se_i^2 \sim N(\mu,\tau^2+\se_i^2),\qquad \tau^2>0
  \label{mod:Gaussian}
\end{align}
That is, conditionally on $\se_1^2,\ldots,\se_\n^2$, $(\Y_1,\ldots,\Y_\n)$ has log-likelihood
\begin{align}
    l(\mu,\tau^2) = -\frac12\sum_{i=1}^n\log(2\pi(\se_i^2+\tau^2)) - \frac12\sum_{i=1}^n\frac{(\Y_i-\mu)^2}{\se_i^2+\tau^2}.
\end{align}
The likelihood ratio (``LR'') statistic for testing $\mu=\munull$ is 
\begin{align}
    \label{defn:W}
  \TLR =  2(l(\mumle,\tausqmle) - l(\munull,\tausqmlered)),
\end{align}
  where $(\mumle,\tausqmle)$ maximizes $l(\mu,\tau^2)$ and  $\tausqmlered$ maximizes
  $\tau^2\mapsto l(\munull,\tau^2)$.

\citep{hardy:96} proposed using the LR statistic so that the variance component $\tau^2$ would be estimated simultaneously with $\mu$, whereas previous methods did not take into account estimation error of the former. The approach also has the benefit of converging at a $1/\n$ rate.

A challenge in carrying out a meta-analysis is that the number of available studies is often small \citep{davey:11}. Efficiency of the statistical procedures is therefore especially valuable. For inference based on the LR, Bartlett correction \citep{bartlett:37,lawley:56,barndorff:88} offers the tantalizing prospect of improving convergence rates from $1/\n$ to $1/\n^2$.

\begin{proposition}\label{prop:Gaussian-Bartlett}
  Suppose model \eqref{mod:Gaussian} holds. Then the Bartlett correction for the LR statistic \eqref{defn:W} is given by
  \begin{align}
    c=1+\frac{2\sum_{i=1}^n\w_i^{3}}{\sum_{i=1}^n\w_i\sum_{i=1}^n\w_i^2} - \frac{\sum_{i=1}^n\w_i^2}{2(\sum_{i=1}^n\w_i)^2},
    \qquad\text{where }w_i=(\se_i^2+\tausqmlered)^{-1}.
  \end{align}
  That is, under model \eqref{mod:Gaussian},    $\TLR/c$ converges to a chi-squared distribution at an $O(1/\n^2)$ rate.
\end{proposition}

See Figure \ref{fig:Gaussian} for an illustration. This formula for
the Bartlett correction corrects that given in \cite{noma:11}, which
omits the third of the three terms, ultimately due to a typo present
in
\cite{cordeiro:93}. 
A crude check on this formula is to compute its limit as $\se_i^2\to 0$,
$i=1,\ldots,\n$, i.e., as the heteroscedastic Gaussian model turns
into the homoscedastic Gaussian model. In this limit the correction
given above becomes $1+3/(2\n)$, matching the well-known Bartlett
correction used in testing for the mean of independent and identically
distributed Gaussian data with unknown variance, e.g., \cite[Example
9.27]{cox:79}. By contrast, the version of the correction proposed in
\cite{noma:11} would imply a correction of $1+2/\n$ for the
homoscedastic Gaussian model.

\begin{figure}
  \centering
  \includegraphics[scale=.5]{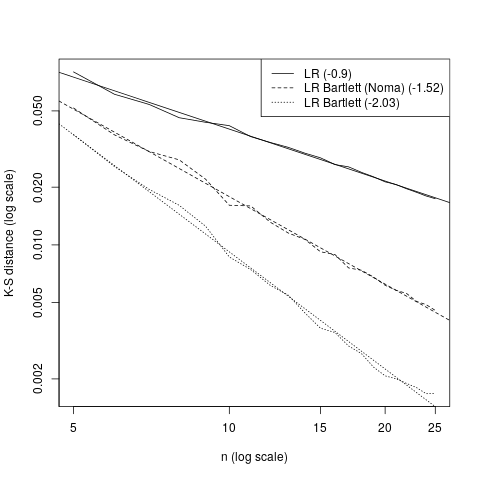}
  \caption{Bartlett correction in the conditional Gaussian model \eqref{mod:Gaussian}. The figure is a log-log plot of the Kolmogorov-Smirnov
    distance between the observed test statistic and the chi-squared
    distribution versus the number of studies $\n=5,\ldots,25$. The
    slopes of the fitted lines given in parentheses approximate the
    convergence rate exponent. The unadjusted LR is approximately $O(1/\n)$ and
    becomes $O(1/\n^2)$ after Bartlett correction. The incomplete
    adjustment of \cite{noma:11} leads to a rate somewhere in
    between.}
  \label{fig:Gaussian}
\end{figure}

 {
\newcommand{\Tc}{T_c}
\newcommand{\sd}{\sigma}
\newcommand{\eps}{\epsilon}
\newcommand{\asym}{dF_{\Tc}(\sqrt{x}+\frac{x^{3/2}}{4\n}) - dF_{\Tc}(-\sqrt{x}-\frac{x^{3/2}}{4\n})}

\begin{proof}[Proof of Proposition \ref{prop:Gaussian-Bartlett}]
  As verified by \cite{noma:11}, the parameters $\mu$ and $\tau^2$ are
  orthogonal under the conditional Gaussian model \eqref{mod:Gaussian}. Therefore the simplified matrix
  approach to Bartlett correction given in \cite{cordeiro:93}
  applies. The intended formula for the correction, which
  \cite{noma:11} relies on, is given by equation (10) of
  \cite{cordeiro:93}. The formula, as written there, is
  \begin{equation}
\begin{split}
        \label{proof:homoscedastic-Bartlett:bad-formula}
    c &= 1 + k_{\alpha,\alpha}^{-1}(\mathbbm{1}^TM\mathbbm{1}+a^TK_{\beta\beta}^{-1}d+\frac{1}{4}(\mathbbm{1}^T(K_{\beta\beta}^{-1}\ast B)\mathbbm{1})^2)\\
    &+k_{\alpha,\alpha}^{-2}(\alpha_1+a^TK_{\beta\beta}^{-1}d+\alpha_2\mathbbm{1}^T(K_{\beta\beta}^{-1}\ast B)\mathbbm{1})
     +k_{\alpha,\alpha}^{-3}\alpha_3.
    \end{split}
      \end{equation}
      Several transcription/typographical errors appear upon comparing equation (10) of \cite{cordeiro:93} with the displays immediately preceding it.  Repairing these, the formula becomes
  \begin{align}
    c &= 1 + k_{\alpha,\alpha}^{-1}(\mathbbm{1}^TM\mathbbm{1}+a^TK_{\beta\beta}^{-1}d+\frac{1}{4}(\mathbbm{1}^T(K_{\beta\beta}^{-1}\ast B)\mathbbm{1})^2)\\
      &+k_{\alpha,\alpha}^{-2}(\alpha_1+a^TK_{\beta\beta}^{-1}b+\alpha_3\mathbbm{1}^T(K_{\beta\beta}^{-1}\ast B)\mathbbm{1})
        +k_{\alpha,\alpha}^{-3}\alpha_2.
  \end{align}
  The material typo is $b\to d$ in the term $a^TK_{\beta\beta}^{-1}b$. As \cite{noma:11} shows, $d=0$ in the Gaussian model, so that term
  does not appear in the formula for the correction if \eqref{proof:homoscedastic-Bartlett:bad-formula} is used. When one substitutes
  \begin{align}
    b = \frac{3}{4}\E(\frac{d^3l(\mu,\tau^2)}{d\mu d\mu d(\tau^2)})
    - \frac{d}{d\tau^2}\E(\frac{d^2l(\mu,\tau^2)}{d\mu d\mu})
    =-\frac{1}{4}\sum_{i=1}^n\frac{1}{(\tau^2+\sigma_i^2)^2},
  \end{align}
  keeping everything else from \cite{noma:11}, one obtains the formula given in the proposition.
  
\end{proof}

\bibliographystyle{chicago}
\bibliography{ma-equivalence.bib}

\end{document}